\documentclass{ecai}  

\usepackage{graphicx}
\usepackage{latexsym}
\usepackage{booktabs}
\usepackage{amsthm}
\usepackage{threeparttable}
\usepackage{algorithm}
\usepackage{tabularx}
\usepackage{xcolor}
\usepackage{algorithmic}
\usepackage{amssymb}
\newtheorem{example}{Example}
\newtheorem{lemma}{Lemma}
\newtheorem{defn}{Definition}
\newtheorem{theo}{Theorem}

\newtheorem{corollary}{Corollary}

\begin{document}

\begin{frontmatter}

\title{Characterizations of Network Auctions and Generalizations of VCG}

\author[a]{\fnms{Mingyu}~\snm{Xiao}\thanks{Corresponding author. Email: myxiao@uestc.edu.cn.}}
\author[a]{\fnms{Guixin}~\snm{Lin}}
\author[a]{\fnms{Bakh}~\snm{Khoussainov}} 
\author[a]{\fnms{Yuchao}~\snm{Song}}

\address[a]{University of Electronic Science and Technology of China, Chengdu, China}

\begin{abstract}
  With the growth of networks, promoting products through social networks has become an important problem. For auctions in social networks, items are needed to be sold to agents in a network, where each agent  can bid and also diffuse the sale information to her neighbors. Thus, the agents' social relations are  intervened with their bids in the auctions. In network auctions, the classical VCG  mechanism fails to retain key properties. In order to better understand network auctions, in this paper, we characterize network auctions for the single-unit setting with respect to weak budget balance, individual rationality, incentive compatibility, efficiency, and other properties. For example, we present sufficient conditions for mechanisms to be efficient and (weakly) incentive compatible. With the help of these properties and new concepts such as rewards, participation rewards, and so on, we show how to design efficient mechanisms to satisfy incentive compatibility as much as possible, and incentive compatibility mechanisms to maximize the revenue. Our results  provide insights into understanding auctions in social networks.
\end{abstract}
\end{frontmatter}

  \section{Introduction}

   Over the last two decades, mechanism design in social networks  has turned into an important research topic. Research on incentivizing information diffusion in social networks is mainly divided into two categories: non-strategic agent settings \cite{kempe2003maximizing} and strategic agent settings \cite{krishna2009auction,emek2011mechanisms}.
   Classical auction design did not incorporate social network referrals into the sale of commodities, which partially limited the seller's revenue or social welfare \cite{borgatti2009network,nisan2007algorithmic,bajari2003winner}. As  opposed   to classical auctions, the strategic transmission of the sale information is the process of expanding the expected results by setting appropriate  tools that  incentivize information diffusion (i.e., in the forms of referrals). Such viral marketing mechanisms are more economically efficient in attracting audiences than advertising in TV, newspapers, and search engines \cite{emek2011mechanisms,leskovec2007dynamics}.

  In the classical auction design, not much attention  was   paid to introducing social network referrals into the sale of commodities \cite{krishna2009auction,myerson1981optimal}. In some respects, this implies that the seller's revenue or social welfare in the classical auction can  be partially optimized. Although the classical mechanism VCG \cite{vickrey1961counterspeculation,clarke1971multipart,Groves1973Incentives} can be applied in network auctions, it may lead to a deficit for the seller.  Consequently, there has been growing interest in designing and analyzing mechanisms to handle auctions where agents are connected through network links \cite{borgatti2009network}.

  To overcome the budget balance problem of the VCG,  a series of studies have been conducted on network auctions. Li et al. \cite{li2017mechanism} presented the information diffusion mechanism (IDM) in social networks, which is not only  budget balance   but  also incentive compatible. Lee \cite{DBLP:conf/sigecom/Lee16} and Jeong and Lee \cite{jeong2020groupwise} then studied budget balance mechanisms to maximize the social welfare at the expense of truthfulness. Takanashi et al. \cite{DBLP:journals/corr/abs-1904-12422} proposed mechanisms to make a trade-off between budget feasibility and efficiency. Li, Hao, and Zhao \cite{li2020incentive} characterized certain conditions for incentive compatible mechanisms. Moreover, there are numerous extended models that follow this research line, such as, network auctions under a multi-unit unit-demand setting studied by \cite{zhao2018selling,DBLP:conf/aaai/KawasakiBTTY20}, multi-unit network auctions with budgets studied by~\cite{DBLP:conf/aaai/XiaoSK22}, and many others \cite{DBLP:journals/corr/abs-2010-04933,li2018customer,zhang2020sybil}. Recent surveys of this area can be found in \cite{DBLP:conf/ijcai/GuoH21,DBLP:journals/ai/LiHGZ22}.

  Although many different mechanisms for network auctions have been proposed, we still lack general theorems  that study the interplay between key concepts in mechanism design for network auctions. Our paper aims to address this gap by scrutinizing single-unit auctions in social networks, with particular attention to the interplay between efficiency and incentive compatibility.
  The results of our investigation will be beneficial in advancing our comprehension and crafting mechanisms for network auctions.

\textbf{Contributions.}  The setting of our work follows the most fundamental and well-studied model~\cite{li2017mechanism,DBLP:journals/ai/LiHGZ22}. A seller desires to sell a commodity to agents connected through a network. We need to incentivize agents to not only report truthful bids but also further propagate the sell information to their neighbors. We want to characterize mechanisms that address the following important properties: efficiency, weak budget balance (WBB), individual rationality (IR), incentive compatibility (IC), and false-name proof (FNP).  Section \ref{S2} defines these concepts and presents our models.  The contributions of our work are enumerated below.

  \noindent
  1. We introduce the concept of \emph{reward} for each agent.  The reward quantifies the agent's influence on the market by attending or abstaining from the auction.
   We show that rewarding agents is important for ensuring incentive compatibility.  Specifically, we prove that in any mechanism that satisfies IR, IC, and efficiency, every agent except the winner should receive a utility at least equal to their respective reward. Please see Theorem~\ref{prop-i-0} for further details.

  \noindent
  2. We define the concept of \emph{participation  reward} for each agent.  This is the amount of social welfare that   the agent can explicitly control. We show that  returning   the participation  reward to agents plays a  crucial   role in  weak   incentive compatibility (WIC). Theorem~\ref{prop-i-1} proves that in any mechanism satisfying IR, WIC, and efficiency,  every agent except the winner should receive a utility that is at least equal to their respective participation reward.

  \noindent
  3.  We design three mechanisms.  The first is the participation  VCG mechanism (PVCG) that   maximizes the revenue  in the class of all  IR, WIC, and efficient mechanisms.
   The second and third are the interruption VCG mechanism (IVCG) and $\delta$-IVCG. They are IC mechanisms at the expense of allocation efficiency. Furthermore, we show that they can get a good revenue.  Specifically, their revenue   is
   not less than that of the most well-known mechanisms on any profile.

  \noindent
  4. Most previous network auction mechanisms assume that the agents report  subsets of their neighbors.
  However,  an agent may create replicas of herself and report a false neighbour set to gain more utility. This behavior   is  known as {\em false-name attacks}. False-name  attacks have   been studied in areas such as multi-level marketing \cite{shen2019multi}, social choice \cite{conitzer2010using}, and blockchains \cite{ersoy2018transaction}.  In network auctions, false-name attacks  can seriously affect the utility  of the seller and agents in a social network \cite{yokoo2004effect}. Zhang et al. \cite{zhang2020sybil} proposed double geometric mechanisms (DGM) that address   false-name attacks and characterize the  uniqueness  of the mechanisms   with  various key properties.  Despite these works, false-name attacks were not well studied on network auctions. Our proposed IVCG is a WBB  and IC mechanism that can also prevent false-name attacks.
  Such mechanisms  are called {\em false-name proof} (FNP) mechanisms.

  Table \ref{tbl}  presents a full comparison among
  mechanisms VCG, IDM, PVCG, IVCG, and $\delta$-IVCG with regard to efficiency, IR, IC, WBB, and FNP properties.

  \begin{table}
    \caption{Comparison among VCG, IDM, PVCG, IVCG, and $\delta$-IVCG}
    \label{tbl}
  \centering
   \begin{threeparttable}
  \scalebox{0.8}{
  \begin{tabular}{cccccccc}
  \toprule
   Mechanisms  &  efficiency & IR &IC &WIC & WBB & FNP   \\
    \midrule
    VCG &\checkmark &  \checkmark &\checkmark &\checkmark  & $\times$ & $\times$  \\
    IDM &$\times$ &\checkmark &\checkmark &\checkmark  & \checkmark &$\times$ \\
    PVCG &\checkmark & \checkmark &$\times$ &\checkmark& \checkmark&$\times$ \\
    IVCG &$\times$& \checkmark &\checkmark & \checkmark & \checkmark& \checkmark  \\
    $\delta$-IVCG &$\times$& \checkmark &\checkmark & \checkmark & \checkmark& $\times$  \\
    \bottomrule
    \end{tabular}
   }
      \end{threeparttable}
  \end{table}

  \section{The Model and Basic Definitions} \label{S2}

  A social network consists of a seller $s$ together with $n$ agents $1$, $2$, $\ldots$, $n$. We sometimes call the agents \emph{bidders} or \emph{buyers}  and denote the set of all agents by $N=\{1, \ldots, n\}$.
  Within the network, there exist undirected edges connecting nodes (agents or seller), which signify either social or business connections between them.   The seller   wants to sell a single commodity in the network. The seller and agents can  exchange sale information only through their neighbors in the network. We call an auction model under this setting the network auction. Our auctions are always network auctions unless otherwise stated.  When the seller is directly connected to all agents, the network auction reduces to the classical auction without networks.

  \smallskip

  For an agent  $i\in N$, let $r_i$ be  the set of neighbors of  $i$. The seller's neighbor set is $r_{s}$. Each agent $i$ has a \emph{private valuation} $v_{i}\ge 0$,   the maximum   value that the agent
   can pay for the commodity, which is also called the \emph{bid} of agent $i$.  At the start of the auction, only the seller's neighbors $r_{s}$ are informed of the sale.

  \smallskip

During the auction, each agent $i \in N$ reports a bid $v_{i}'$ and a subset of neighbors $r_{i}'$ who are informed of the sale. Note that both the bid $v_{i}'$ and the neighbour report $r_i'$ can be different from the truthful reports $v_{i}$ and $r_i$.  Call the pair $(v'_{i},r'_{i})$ the  \emph{auction profile} (or \emph{profile} for short) of  agent $i$ and denote it  by $a_i'$. If  $a_i'$ equals $(v_{i},r_{i})$, then   agent $i$ is \emph{truthful}. The  profile $a'_{i}=\emptyset$ means that  agent $i$ does not want to attend the auction or agent $i$ is not aware of the sale. A  \emph{global  profile} is the  vector $a'=(a_{i}', i\in N)$ for all auction profiles. By $a_{-i}'=(a_{j}', j\neq i)$, denote the  profile of all agents but $i$.

  \smallskip

  Given a  global profile $a'$, we can construct a \emph{spreading graph}. The set of nodes of the spreading graph consists of  seller $s$ and all agents. For  two agents $i$ and $j$, there is a directed edge from $i$ to $j$ if $j\in r'_i$.  Hence, there is a directed edge from $s$ to each agent in $r_s$. There may be a pair of opposite edges between two agents.
  The spreading graph constructed based on the global profile may be different from the existing social network even we ignore the edge direction in the spreading graph.
  We will say agent $i$ is \emph{closer} to $s$ than agent $j$ if there is a directed path from $s$ to agent $i$ with the length not greater than that of any directed path from $s$ to $j$ in the spreading graph.
  An agent $i$ is considered \emph{activated} if there exists a directed path from the seller $s$ to $i$ in the spreading graph, whereas \emph{unactivated} agents are not informed of the sale.
  We ``trim'' the global  profile $a'$ and spreading graph by letting $a_i'=\emptyset$ for all unactivated agents $i$ and deleting all unactivated agents from the spreading graph. We may not draw a directed edge from $j$ to $i$ if there is  an edge from $i$ to $j$ and the distance from  $s$ to $j$ is greater than the distance from  $s$ to $i$. This kind of edge $ji$ will not affect any property since telling the sale information to an agent already knowing the sale is useless in our model. However, readers can simply assume this kind of edge is existing even it is not drawn.



  \begin{defn} (\textbf{Mechanism})
      An \emph{auction mechanism} $\mathcal{M}$, given a global profile $a'$, outputs an
      \emph{allocation policy} $\pi: N\rightarrow \{0,1\}$,  and  a \emph{payment policy} $m: N\rightarrow R$. Since there is only one item to sell, we  have $\sum_{i\in N}\pi(i)=1$.
  \end{defn}

  Let $\mathcal M$ be  a mechanism.   Agent $i$ with $\pi(i)=1$ is the \emph{winner} who gets the commodity.  For the payment policy, if $m(j)\geq 0$ then agent $j$ pays $m(j)$, and if  $m(j)< 0$ then agent $j$ receives $|m(j)|$.  A mechanism $\mathcal{M}$ is \emph{feasible} if it allocates the commodity to an activated agent. From now on, mechanisms are always feasible.

  \smallskip
  Given a profile $a'$, the utility of agent $i$ under a mechanism $\mathcal{M}$ is denoted by \ $u_i(a', \mathcal{M}):= \pi(i)v_i-m(i)$.  Meanwhile, the utility $u_s(a', \mathcal{M})$ of  seller $s$, also called the \emph{revenue}, is  the sum of the payments: $m(1)+\ldots + m(n)$.
  The sum of utilities of all participants in the social network is called the \emph{social welfare} and denoted by  $SW(a', \mathcal{M})$. Thus:
  $$
  SW(a', \mathcal{M}) = \sum_{i\in N\cup \{s\}}   u_i(a', \mathcal{M})=\sum_{i\in N} \pi(i)v_i.
  $$
We also define the \emph{allocation efficiency} as
$$
AE(a', \mathcal{M}):=\sum_{i\in N} \pi(i)v'_i.
$$
We have $SW(a', \mathcal{M}) \in \{v_1, \ldots, v_n\}$ and $AE(a', \mathcal{M}) \in \{v'_1, \ldots, v'_n\}$.

\begin{defn} (\textbf{Social Welfare Maximum} and \textbf{Efficiency})
  A mechanism $\mathcal{M}$ is \emph{social welfare maximizing (SWM)} if $SW(a', \mathcal{M})\geq SW(a', \mathcal{M}')$ holds for any global  profile $a'$ and any mechanism $\mathcal M'$;
  A mechanism $\mathcal{M}$ is \emph{efficient} if $AE(a', \mathcal{M})\geq AE(a', \mathcal{M}')$ holds for any global  profile $a'$ and any mechanism $\mathcal M'$.
  \end{defn}




  \begin{defn} (\textbf{Weak Budget Balance})
  A mechanism $\mathcal{M}$ is \emph{weakly budget balanced (WBB)} if for any  profile $a'$ the utility of the seller is non-negative:  $u_s(a', \mathcal{M})\geq 0$.
  \end{defn}

   \begin{defn} (\textbf{Individual Rationality})
  A mechanism $\mathcal{M}$ is \emph{individually rational (IR)} if for each agent $i$, her utility is non-negative as long as she reports her bid and neighbors truthfully, i.e., $u_i((a_i, a'_{-i}), \mathcal{M})\geq 0$ holds for any
  $a'_{-i}$.
  \end{defn}


  \begin{defn}(\textbf{Incentive Compatibility})
  A mechanism $\mathcal{M}$ is \emph{incentive compatible (IC)} if each agent $i$ gets the highest utility by reporting her  profile truthfully, i.e., $u_i(((v_i, r_i), a'_{-i}), \mathcal{M})\geq u_i((a'_i, a'_{-i}), \mathcal{M})$ holds for any profiles $a_i'$ and $a'_{-i}$.
  \end{defn}

  In the definition of incentive compatibility, we want to prohibit agents  from misreporting both  higher and lower bids  than their true values for the commodity. However, sometimes this is too strong for some good mechanisms to satisfy. Note that false lower bids may be more harmful to society than false higher bids in some cases, and higher bids may be constrained by the budget and others. Thus, in the literature a weaker version of incentive compatibility, which prohibits false lower bids only, has been studied (see~\cite{goel2014mechanism}).
  It is also called \emph{one-way truthfulness}.

  \begin{defn} (\textbf{Weak Incentive Compatibility})
  A mechanism $\mathcal{M}$ is \emph{weakly incentive compatible (WIC)}
  if $u_i(((v_i, r_i), a'_{-i}), \mathcal{M})\geq u_i(( (v'_i, r'_i), a'_{-i}), \mathcal{M})$ holds for any profiles
  $a'_{-i}$ and $(v'_i, r'_i)$ such that $v'_i\leq v_i$ and $r_i'\subseteq r_i$.
  \end{defn}

Next, we are ready to consider sufficient conditions for mechanisms satisfying the above properties.

  \section{Social Welfare and Efficiency}
  Social welfare is an important concept frequently considered in mechanism design.
  However, controlling social welfare can be challenging.  We give an example to explain this.
  There are two agents with reported bids 1 and 100. Their truthful bids can be (1,100) and (100,1).
  In the first case, allocating the commodity to the second agent maximizes social welfare. In contrast, allocating the commodity to the first agent yields maximum social welfare in the second case. Unfortunately, it is impossible to determine which agent to allocate based solely on reported bids, rendering SWM impractical.  Since SWM is too strong to satisfy independently, we consider efficiency.
  The agent with the highest bid among all activated agents is called the \emph{highest bidder}.
  According to the definition of social welfare and efficiency, we can easily derive the following property.
%
%
  \begin{lemma}\label{lemma1-effi}
  An efficient mechanism always allocates the commodity to the highest bidder.
  Furthermore, if all agents report truthfully, an efficient mechanism reaches the maximum social welfare. \qed
  \end{lemma}

  To design efficient mechanisms, we only consider the payment policy since the allocation policy is fixed. We use the indices $\{i_1, i_2, \dots, i_n\}$ to rank the agents such that $v'_{i_j}\geq v'_{i_k}$ if and only if $j\leq k$. Thus, agent $i_1$ is always the highest bidder.

  We now consider how an agent influences the allocation efficiency   by not attending an auction.
  Let $a'$ be  a global  profile.
  Let $H(a')$ be the highest bid in $a'$.
 We have that  $H(a')=v'_{i_1}$.  Let $a'_{-[i]}$ be the global  profile obtained from $a'$ by (1) replacing $a'_i$ with $\emptyset$ indicating that agent $i$ does not attend the auction, and (2) replacing the auction  profiles   of all thus unactivated agents (resulted by (1)) with $\emptyset$.

  \begin{defn} (\textbf{Rewards})
  For a global profile $a'$,
  the \emph{reward} of  agent $i$ is defined as \  $rwd_i:=H(a')-H(a'_{-[i]})$.
  \end{defn}

The rewards present agents' contribution to the efficiency by participation and diffusion.
  The next result sheds light on the relationship between rewards and key properties IR, IC, and efficiency.

   \begin{theo}\label{prop-i-0}
   In network auctions,  any mechanism satisfying IR, IC, and efficiency possesses the following properties  :
  \begin{enumerate}
  \item[(a)] The highest bidder $i_1$ gets the commodity with a payment not greater than $H(a'_{-[i_1]})$.
  \item[(b)] The utility of any agent $j$ apart from  the highest bidder $i_1$ is at least her reward $rwd_j$.
  \end{enumerate}
   \end{theo}

   \begin{proof}
  Part $(a)$:  The highest bidder  $i_1$  is  awarded the commodity as the mechanism is efficient. Assume that agent $i_1$ pays more than $v^*=H(a'_{-[i_1]})$,  the highest bid in $a'_{-[i_1]}$.  Hence, the payment of $i_1$ is $v^*+\varepsilon$ for some $\varepsilon> 0$. We have $v'_{i_1} \geq v^*+\varepsilon$, otherwise the mechanism is not  IR. The utility of $i_1$ is $v_{i_1}-v^*-\varepsilon$. If so,  agent $i_1$ gets more utility by reporting her bid as $v^*+\varepsilon/2$ instead of $v'_{i_1}$ and setting $r_{i_1}'=\emptyset$. By doing so,  agent $i_1$ is still the unique highest bidder who gets the commodity. Moreover, the payment of $i_1$ should not be greater than $v^*+\varepsilon/2$. Otherwise,  the mechanism is not IR if $v^*+\varepsilon/2$ is the truthful evaluation of agent $i_1$.  So the utility of $i_1$ will be at least $v_{i_1}-v^*-\varepsilon/2$, greater than $v_{i_1}-v^*-\varepsilon$. This contradicts the assumption that  the mechanism is IC. We proved Part $(a)$.

  \smallskip
  \noindent
  Part $(b)$: Consider an agent $j$, where $j\neq i_1$.
Since the mechanism is IR, we know that no agent will get a negative utility.
So, if $rwd_j=0$, then we are done.  Assume that $rwd_j>0$.
If the utility of  agent $j$ is $rwd_j-\varepsilon$ for some $\varepsilon>0$, less than  $rwd_j=H(a')-H(a'_{-[j]})=v'_{i_1}-H(a'_{-[j]})$,
then  agent $j$ would get more utility by reporting her bid as $v'_{i_1}- \varepsilon/2$. The reason is as follows.
Let $a''$ be the profile obtained from $a'$ by replacing agent $j$'s bid with $v'_{i_1}- \varepsilon/2$.
For the profile $a''$ (we consider it as truthful),  agent $j$ must get utility at least $(v'_{i_1}- \varepsilon/2)-H(a'_{-[j]})$,  because if she does not get so much she may falsely report $r_j'=\emptyset$ to become the highest bidder and get this utility by Part (a).
Thus,  agent $j$ can falsely report her bid to $v'_{i_1}- \varepsilon/2$ to get utility at least $rwd_j-\varepsilon/2$, more than $rwd_j-\varepsilon$.
   \end{proof}

  \noindent

  Observe that if a mechanism tries to reward each agent, it may be too much for the mechanism to be WBB.
  Namely, we have the following result  that was  also observed by   ~\cite{li2020incentive}.

   \begin{corollary}\label{hardresult}
   In network auctions, no IR and WBB mechanism satisfies both of IC and efficiency. \qed
   \end{corollary}

   \noindent

  We give a simple example that explains Corollary~\ref{hardresult}.

   \begin{example}{\em
   
   The example contains only two agents $\{a,b\}$. The spreading graph is a path from seller $s$ to $a$ and then to $b$.
    Suppose that agent $a$ submits a bid of 0, while agent $b$ submits a bid of 1.
     Assume that the mechanism is  IR, IC and efficient. By Theorem \ref{prop-i-0}(a), agent $b$ gets the commodity and pays at most 0.
   By Theorem \ref{prop-i-0}(b), agent $a$ gets at least $1-0=1$. The revenue  is $0-1=-1< 0$. Thus, the mechanism is not WBB.}
   \end{example}

To better explain the mechanisms, we first recast the famous Vickrey-Clarke-Groves (VCG) mechanism \cite{vickrey1961counterspeculation,clarke1971multipart,Groves1973Incentives}. 
We extend the VCG to the network setting and present the mechanism by using the new concept of rewards.

  \begin{quote}{ \textbf{The  VCG Mechanism in Networks}:\\
  Allocate the commodity to  agent $i_1$ with the highest bid $v'_{i_1}$ (break ties arbitrarily), charge  agent $v'_{i_1}$, and pay all agents $i$, including $i_1$, the reward $rwd_i$.}
  \end{quote}
 
 The VCG mechanism in networks sells the commodity to the highest bidder $i_1$  and returns agents the reward they deserve, which is essentially the same VCG mechanism in \cite{li2017mechanism,DBLP:conf/sigecom/Lee16} and the IC property has been proved. The payment policy of the VCG for each agent $i$ is defined as $m(i)=\pi(i)v'_{i}- rwd_i$. When there is no network, the reward of the highest bidder $i_1$ is
  $rwd_{i_1}=v'_{i_1}-v'_{i_2}$ and the reward of all other agents is $0$.  The payment of the highest bidder $i_1$ is  $m(i_1)=\pi(i_1)v'_{i_1}- rwd_{i_1}=v'_{i_1}-(v'_{i_1}-v'_{i_2})=v'_{i_2}$,
  the second highest bid. All other agents pay 0. For this case, the mechanism degenerates to the normal VCG.
  By Theorem \ref{prop-i-0}, we have:

  \begin{corollary}
  In network auctions, although the VCG in networks is not WBB, it maximizes the revenue  among all mechanisms satisfying IR, IC, and efficiency. \qed
  \end{corollary}


  \begin{defn}
      An agent $i$ is called \emph{critical} if $rwd_i>0$. If after deleting critical agent $i$ there is no path from the seller $s$ to another critical agent $j$, then we say agent $i$ is agent $j$'s \emph{critical ancestor}, and agent $j$ is agent $i$'s \emph{critical descendant}.
  \end{defn}


  \begin{example}{\em
  In Figure \ref{G21}, we have a spreading graph with $11$ agents. For each node, the letter in the cycle is the name of the agent and the number beside  the cycle is the bid of the agent. Agents  $b, e, j$ and $k$ in Figure \ref{G21} are critical agents.
  }
  \end{example}

  Critical agents influence the allocation efficiency and we will pay  special attention to them.

  \begin{figure}[t]
      \centering
      \includegraphics[scale=0.7]{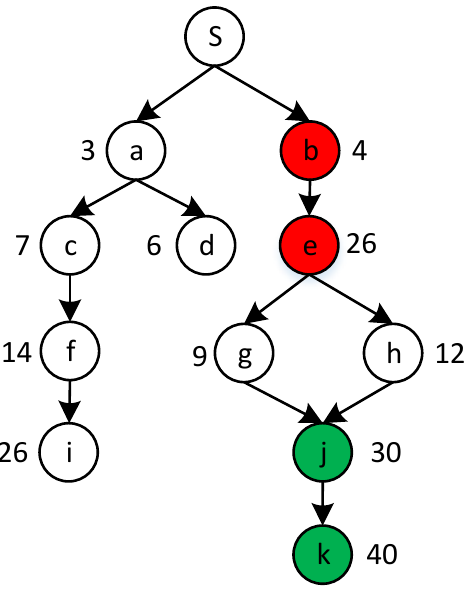}
    \caption{An example consists of seller $s$ and 11 agents. Red and green agents are critical agents, green agents are also interruption agents, and the numbers beside the agents are the bids.}\label{G21} 
  \end{figure}
  \noindent

  \begin{lemma}\label{prop-path}
  If there is only one highest bidder, then the bidder is a critical agent and all other critical agents (if they exist) are the highest bidder's critical ancestors.
  When there is more than one highest bidder, critical agents might not exist.
  \end{lemma}
  \begin{proof}
      Assuming a single highest bidder $i_1$,  this bidder is inherently a critical agent.
      For any critical agent, if she does not attend the auction, the highest bid in the auction will change. Thus, the deletion of it will disconnect  the highest seller from the seller in the spreading graph.
      The critical agent is the highest bidder's critical ancestor.

      For the second part of the lemma, we assume that here are two highest bidders and none is a critical ancestor of the other. For this case, the reward of any agent is 0 and there is no critical agent.
  \end{proof}

%

  A critical agent can influence the allocation efficiency in several ways.  For example, the agent can attend the auction without reporting neighbors or can ignore the auction without attending it.
  Doing so  may affect the IC property of mechanisms.

  Let  $a'_{-(i)}$ be the global profile obtained from a given global profile $a'$ by (1)  replacing the profile $a'_i=(v'_i, r'_i)$ with $(v'_i,\emptyset)$ and then (2) replacing the profiles of all unactivated agents with $\emptyset$. It is easy to see the following inequalities: $$H(a'_{-[i]}) \leq H(a'_{-(i)}) \leq H(a').$$

  \begin{defn}(\textbf{Participation  Reward})
  The \emph{participation  reward} for an agent $i$ is defined as
  \begin{equation}\label{prwd}
  prwd_i:=H(a'_{-(i)})- H(a'_{-[i]}).
  \end{equation}
  If $prwd_i>0$, then agent $i$ is called an
   \emph{interruption} agent, e.g., in Figure \ref{G21}, agents $j$ and $k$ are interruption agents.
   \end{defn}

  The participation rewards present agents' contribution to efficiency by participation without diffusion.

   \begin{lemma}\label{subset}
  Any interruption agent is a critical agent. Furthermore, an interruption agent will become the unique highest bidder if she reports no neighbors.
   \end{lemma}

   \begin{proof}
   We have $rwd_i \geq prwd_i$ for any agent $i$. An interruption agent $i$ holds that $prwd_i>0$. Thus, $rwd_i> 0$ and then it is a critical agent.
   According to the definition of interruption agent, we know that the highest  bid for the following two cases are different: an interruption agent does not attend the auction; the agent attends the auction without reporting any neighbors. So, the agent herself should be the unique agent giving the highest bid when she reports no neighbors.
   \end{proof}

  \noindent
  Next, we   provide an insight into interaction between  participation  rewards  and  the properties IR, WIC, and efficiency.

  \begin{theo}\label{prop-i-1}
   In network auctions,   any mechanism satisfying IR, WIC, and efficiency possesses the following properties:
  \begin{enumerate}
  \item[(a)] The highest bidder $i_1$ gets the commodity with a payment not greater than $H(a'_{-[i_1]})$.
  \item[(b)]  The utility of each other agent $j\neq i_1$ is at least her participation  reward $prwd_j$.
  \end{enumerate}
  \end{theo}

  \begin{proof}
   The proof of Part (a)  is the same as the proof of  Theorem \ref{prop-i-0}(a). We prove Part (b). The mechanism is IR. So, no agent gets a negative utility. If $prwd_j=0$, then the utility of $j$ is at least $0$ due to IR.  Assume that $prwd_j>0$. If the utility of $j$ is smaller than  $prwd_j$,  then  the agent gets more utility by misreporting her neighbor set $\emptyset$.  By Lemma \ref{subset}, we know that  agent $j$ will become the unique highest bidder. She will pay the commodity with the payment at most $H(a'_{-[j]})$ by the first part of the theorem.  Now the utility of $j$ is at least $v'_{j}-H(a'_{-[j]})=H(a'_{-(j)})- H(a'_{-[j]})=prwd_j$.
   \end{proof}

  Compared to  Theorem \ref{prop-i-0}, Theorem \ref{prop-i-1} states   that we can pay less to agents  in case we aim to ensure the WIC property rather than the IC property. The next proposition  implies that paying the participation  rewards  only  may be good enough to be weakly budget balanced.

   \begin{lemma}\label{prop-i-2}
   Let $v'_{i_1}$ be the highest bid in the global  profile $a'$. Then we have \ $prwd_1+\ldots + prwd_n\leq v'_{i_1}$.
   \end{lemma}

   \begin{proof}
  It suffices to  consider agents $i$ such that $prwd_i> 0$.  Set $N'=\{i \in N \mid prwd_i> 0\}$. By Lemmas \ref{prop-path} and \ref{subset}, all agents in $N'$ are critical ancestors of $i_1$. We list all the agents in $N'$ as $j_1, j_2, \dots, j_x$ such that $j_x=i_1$ and for each $i\in \{1,\dots,x-1\}$, 
    the distance from  $s$ to $j_{i}$ is smaller than that from  $s$ to $j_{i+1}$ in the spreading graph.
    So under the global  profile $a'$, after deleting agent $j_i$, all agents $j_{i'}$ with $i'> i$ will be disconnected from  seller $s$ in the spreading graph.   We have
  \begin{equation} \label{equ3}
  H(a'_{-(j_i)})\leq H(a'_{-[j_{i+1}]})
  \end{equation}
  for each $i\in \{1,\dots,x-1\}$.
  By (\ref{prwd}) and (\ref{equ3}), we have that
  $$ \sum_{i\in N} prwd_i=\sum_{i\in N'} prwd_i\leq H(a'_{-(j_x)})-H(a'_{-[j_1]})$$
  $$=v'_{i_1}-H(a'_{-[j_1]}).$$
  The lemma is proved.
   \end{proof}


  The previous analyses lead us to the following mechanism.

  \begin{quote}{
  \textbf{The Participation  VCG Mechanism (PVCG)}:\\
  Allocate the commodity to  agent $i_1$ with the highest bid $v'_{i_1}$ (break ties arbitrarily), charge  agent $v'_{i_1}$, and pay all agents $i$, including $i_1$, the participation reward $prwd_i$.}
  \end{quote}

  \noindent
  The main difference between VCG and PVCG mechanisms is that the VCG mechanism   pays each agent her reward while the PVCG mechanism pays each agent her participation  reward.   With no network, PVCG and VCG are coincide.

  \begin{theo}\label{th-PVCG}
  In network auctions, the PVCG mechanism is IR, WBB, WIC, and efficient.
  \end{theo}
  \begin{proof}
      The PVCG allocates the commodity to an agent with the highest bid. By Lemma \ref{lemma1-effi}, the mechanism is efficient. By the inequalities $H(a'_{-[i]}) \leq H(a'_{-(i)}) \leq H(a')$, no agent has a negative utility. Hence,  the IR   property holds. The utility of the seller is $v'_{i_1}- \sum_{i\in N} prwd_i$. By Lemma \ref{prop-i-2}, the PVCG mechanism   is WBB.  So,  we need to prove  WIC.

      If  $i$  is not an interruption agent, then her utility is 0. By decreasing  her bid $v'_i$, she cannot become an interruption agent. Hence, her utility will not increase. The utility of the highest bidder $i_1$ is $v_{i_1}-H(a'_{-[i_1]})$. Decreasing her bid $v'_{i_1}$ will change neither $v_{i_1}$ nor $H(a'_{-[i_1]})$. So $i_1$ cannot get more utility by decreasing her bid.  For an interruption agent $i\neq i_1$ that does not get the commodity, her utility is $prwd_i=H(a'_{-(i)})- H(a'_{-[i]})$.   By decreasing $v'_{i}$ she lowers $H(a'_{-(i)})$ without changing $H(a'_{-[i]})$.

      It is apparent that no agent can change her type by misreporting her neighbors. For non-interruption agents, their utility stays constant at 0, so there is no incentive to lie. For interruption agents, neither $H(a'_{-(i)})$ nor $H(a'_{-[i]})$ can be modified through neighbor misreporting. Thus, truthful reporting of neighbors is  a dominant strategy for all agents. We conclude that the PVCG mechanism is WIC.
  \end{proof}

  Theorem~\ref{prop-i-1} describes the minimum utility of each agent  for    mechanisms satisfying IR, WIC, and efficiency.  Theorem~\ref{th-PVCG} designs a mechanism exactly satisfying the minimum requirement. Hence, we derive the following statement:

  \begin{corollary}\label{PVCG-utility}
  In network auctions, the PVCG mechanism maximizes the revenue  among all mechanisms satisfying IR, WIC, and efficiency.
  \end{corollary}

  \textbf{Remark:} There are two known IR, WIC, and efficient mechanisms for network auctions: one is the multilevel mechanism \cite{DBLP:conf/sigecom/Lee16} and the other is GPR \cite{jeong2020groupwise}.  The payment policies of these two mechanisms are more complex than that of PVCG. However, by Corollary~\ref{PVCG-utility},  the PVCG guarantees  the seller gets more revenue.

  \section{Incentive Compatibility}
Next, we  emphasize another key property in mechanism design -- the incentive compatibility.
  The VCG can guarantee both IC and efficiency. However, it is not WBB.  Corollary~\ref{hardresult} says that it is impossible for an IC and IR mechanism to further satisfy both of WBB and efficiency. 
However, it is also important to consider the efficiency in IC mechanisms.
If we do not care the revenue or efficiency, an IC mechanism may do nothing meaningful, say assigning the commodity to an agent randomly  and charging 0 for all agents. We want to get some lower bound on the efficiency. 
First of all, we have a negative result on approximating the efficiency.
{ \begin{theo}
  In network auction no IC and WBB mechanism can archive a constant  ratio on efficiency in worst-case analysis.
\end{theo}
\begin{proof}Let $\mathcal{M}$ be a mechanism with a constant ratio $\alpha$ of the efficiency.
We prove the theorem by giving an example. The example contains only two agents $\{a,b\}$. The spreading graph is a path from seller $s$ to $a$ and then to $b$.
The true valuations of $a$ and $b$ are $0\leq x<\alpha$ and $1$, respectively.
The mechanism $\mathcal{M}$ should always allocate the commodity to $b$ to satisfy the constant ratio property and charge her at most $x$ to satisfy the IC property (similar to the proof of Theorem \ref{prop-i-0}(a)).
However, $\mathcal{M}$ should also give at least $\alpha$ to $a$ to satisfy the IC property, otherwise if $a$ only gets $\alpha'\in (x,\alpha)$, she will misreport to $\alpha''\in (\alpha',\alpha)$ to get more utility.
Now, the revenue is $x-\alpha<0$, which violates the WBB property. Hence, we have proved the theorem.
\end{proof}

When there is no network, the VCG can guarantee the revenue of at least the second highest bid. For network auctions, it is even not easy to get a similar bound for IC mechanisms.

\begin{defn}
Assume that there is at least one interruption agent. Call an interruption agent  the \emph{leading agent} if it  is the closest interruption agent to the seller in the spreading graph.
\end{defn}

We show that for network auctions, there is an IC mechanism that can always get the revenue at least as the highest bid in the auction where the leading agent $j_1$ does not attend, i.e.,  $H(a'_{-[j_1]})$. The mechanism is as below.

\begin{quote}{{\bf The Interruption VCG Mechanism (IVCG)}:\\
If there is no interruption agent, allocate the commodity to an agent $i_1$ with the highest bid $v'_{i_1}$ (break tie arbitrarily) and charge her $v'_{i_1}$; other agents pay 0. Otherwise, allocate the commodity to the leading agent $j_1$ and charge her $v'_{j_1}-prwd_{j_1}=H(a'_{-[j_1]})$; other agents pay 0.}
\end{quote}

\begin{theo}
The IVCG mechanism is IR, WBB, and IC.
\end{theo}
\begin{proof}

The utility of the leading agent $j_1$ (if it exists) is $hrwd_{j_1}$ when she reports the truth, which is nonnegative. The utility of each other agent is 0.
The utility of the seller is either $v'_{i_1}\geq 0$ or  $H(a'_{-[j_1]})\geq 0$. So, the IVCG is IR and WBB.

Now, we consider the IC property. Among all the agents, only the leading agent (if it exists) may have a non-zero utility in the IVCG mechanism. Let $i$ be an agent such that she is not a leading agent when she reports truthfully. Thus, her truthful evaluation $v_i$ is not greater than $H(a'_{-[i]})$.
If agent $i$ becomes a leading agent via misreporting a bid greater than $H(a'_{-[i]})$, then $i$ gets the commodity with a price $H(a'_{-[i]})$, and then her utility is non-positive.
So agent $i$ will not be  a leading agent by misreporting her bid no matter there exists a leading agent or not in the original global profile.

Next, assume that $i$ is the leading agent when she reports truthfully. As long as agent $i$ is the leading agent by reporting any bid, her payment is  $H(a'_{-[i]})$ and her utility is $v_i-H(a'_{-[i]})$ without increment. If agent $i$ becomes not a leading agent by misreporting a bid, then her utility becomes 0, also not increasing.
So any agent can not get more utility by misreporting her bid. Furthermore, it is easy to see that no agent $i$ can change her utility by changing her reported neighborhood $r'_i$. Thus,  the IVCG mechanism   is IC.
\end{proof}

The IVCG can always get a revenue not less than that of some well-known IC mechanisms. Here we compare IVCG with IDM \cite{li2017mechanism} and DAN-MU\footnote{The original DAN-MU was designed for the multi-unit setting and it was later found not to be IC~\cite{DBLP:journals/corr/abs-2208-14591}. However, when there is only one commodity, DAN-MU is IC. Next, when we mention DAN-MU we always mean DAN-MU with one commodity.} \cite{DBLP:conf/aaai/KawasakiBTTY20} as examples.
We first give an example to show that the revenue in IVCG is greater than that in IDM and DAN-MU. The spreading graph of the example is shown in Figure \ref{fig1} and the running results of the three mechanisms  are  shown in Table \ref{tb1}. In the table, the number is the payment of the corresponding agent and the payment  with `$\star$'  means the corresponding agent is the winner. The revenue in IVCG is 2 and the revenue in IDM and DAN-MU is 1.
Next, we theoretically prove that IVCG can always get a revenue not less than that of IDM and DAN-MU on any input profile.

\begin{figure}[t]
  \centering
  \includegraphics[width=0.13\textwidth]{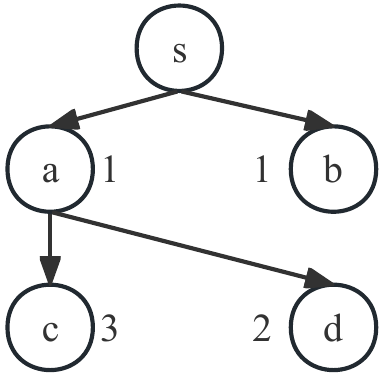}
  \caption{The spreading graph of an instance with a seller $s$ and 4 buyers $a$, $b$, $c$, and $d$. The numbers beside the agents are the bids}
  \label{fig1}
\end{figure}

\begin{table}[t]
  \centering
  \caption{Running results of IVCG, IDM, and DAN-MU on the instance in Figure \ref{fig1}. The winners are marked by `$\star$'.}
  \label{tb1}
  \begin{tabular}{lllll}
      \hline
      Agents  & IVCG   & IDM &DAN-MU  \\
      \hline
      a &0            &-1          &$1^{\star}$ \\
      b &0            &0           &0 \\
      c &$2^{\star}$  &$2^{\star}$ &0     \\
      d &0            &0           &0 \\
      \hline
      Revenue &2 &1  &1\\
      \hline
  \end{tabular}

\end{table}

\begin{lemma}
  The revenue  in  the IVCG mechanism is not smaller than that of the IDM and DAN-MU mechanisms.
  \end{lemma}

  \begin{proof}
      When there is no interruption agent,  the IVCG allocates   the commodity to one of the highest bidders and the revenue is the highest bid, which is
      the maximum revenue  under all IR mechanisms. Next, we assume that there exists an interruption agent.

      In  the IVCG mechanism, the revenue  is $H(a'_{-[p]})$, where $p$ is the leading agent who gets the commodity.
      In  the IDM, the revenue  is $\sum_{i=1}^{x-1}(H(a'_{-[c_i]})-H(a'_{-[c_{i+1}]}))+H(a'_{-[w]})$, where $c_x=w$ is the agent who gets the commodity.
      So, in  the IDM, the revenue is $H(a'_{-[c_1]})$.
       Thus, we only need to compare $H(a'_{-[p]})$ and $H(a'_{-[c_1]})$. When there exists a leading agent, both of $p$ and $c_1$ are cut-vertices on any path from  seller $s$ to the highest bidder in the spreading graph. Furthermore, $c_1$ is the first cut-vertex with the shortest distance to the seller. So we know that $H(a'_{-[c_1]}) \leq H(a'_{-[p]})$. We conclude that the revenue in the IVCG is not less than that of the IDM.

       The DAN-MU mechanism  allocates the commodity to the closest agent $w$ to the seller such that $v_{w}'$ is the highest bidder among all bidders except her critical descendants.
       If $w$ is the unique highest bidder after she reports the neighbor set as $\emptyset$, agent $w$ is also the winner in the ICVG.
       Otherwise, the participation reward of $w$ is 0 and the leading agent $p$ will be
       a critical descendant of $w$.
       We always have that $H(a'_{-[p]}) \ge H(a'_{-[w]})$.
      \end{proof}

An interesting question is whether the revenue in any IC mechanism can not break the barrier of $H(a'_{-[j_1]})$ for any global profile $a'$ with a leading agent $j_1$.
Unfortunately, the answer is negative. To understand this well, we designed a special IC mechanism, called $\delta$-IVCG that can get $H(a'_{-[j_1]})+\delta$ for any $\delta >0$ on some profiles.

The $\delta$-IVCG is the same as IVCG except for one case. Let $\delta$ be a non-negative value.
We say that a global  profile $a'$  is \emph{$\delta$-gap} if it holds that: \\
(a) agent $j_1$ is the only neighbor of seller $s$;\\
(b) there is a critical descendant $j_2$ of $j_1$ such that $j_2$ is an interruption agent;\\
(c) it holds that $H(a'_{-[j_2]})-v'_{j_1}\geq \delta$. \\
W.l.o.g., for the sake of presentation, we always assume that $j_2$ is the closest agent to the seller satisfying the above conditions.

\begin{quote}{{\bf The $\delta$-IVCG Mechanism}:\\
If the global profile is not $\delta$-gap, then do the same as the IVCG. Else the global profile is $\delta$-gap,
the mechanism allocates the commodity to $j_2$ and charges her $H(a'_{-[j_2]})$, gives $j_1$ money $H(a'_{-[j_2]})-\delta-H(a'_{-[j_1]})$, and charges all other agents 0.
}
\end{quote}

\begin{theo}
  The $\delta$-IVCG mechanism is IR, WBB, and IC.
\end{theo}
\begin{proof}
  We have proven that the IVCG is IR and WBB. Only on $\delta$-gap profiles the operations of the $\delta$-IVCG are different from that of the IVCG. Next, we prove the IR and WBB properties of the $\delta$-IVCG on $\delta$-gap profiles. Consider the $\delta$-IVCG on a  $\delta$-gap profile $a'$.
  The utility of agent $j_1$ is
  $H(a'_{-[j_2]})-\delta-H(a'_{-[j_1]})=H(a'_{-[j_2]})-\delta>0$ since $H(a'_{-[j_1]})=0$,
  the utility of agent $j_2$ is $v_{j_2}-H(a'_{-[j_2]})>0$,
   and the utility of other agents is 0. Hence, the $\delta$-IVCG is IR. The revenue of the seller is
  $H(a'_{-[j_2]})-(H(a'_{-[j_2]})-\delta-H(a'_{-[j_1]}))=H(a'_{-[j_1]})+\delta>0$.
   Thus, the $\delta$-IVCG is WBB. Next, we consider the IC property.

     We first prove that the $\delta$-IVCG is IC on $\delta$-gap profiles.
     For a $\delta$-gap profile $a'$, according to the utility of the agents, we have three different types:  agent $j_1$,  winner $j_2$,
      and other agents (called normal agents). First, we consider  agent $j_1$.
      Agent $j_1$ is always the leading agent and her utility is $H(a'_{-[j_2]})-\delta$. If
      the profile is still $\delta$-gap after her misreporting,
     she can not get more utility since she can not change the two values $H(a'_{-[j_2]})$ and $\delta$.
     If
     the profile is not $\delta$-gap after her misreporting,
     she can not get more utility since her utility will become $v_{j_1}$, which is not greater than $H(a'_{-[j_2]})-\delta$
     by the definition of $\delta$-gap profile.
      So agent $j_1$ will not misreport.
     For the winner $j_2$, her utility is $v_{j_2}-H(a'_{-[j_2]})$. If she is still the winner after misreporting,
     she can not get more utility since she can not change her truthful evaluation $v_{j_2}$ and $H(a'_{-[j_2]})$.
     If she is not the winner after misreporting, she can only be a normal agent since she can not become a leading agent for any case and then her utility is 0. Thus, agent $j_2$ will not misreport.
     Note that in the definition of $\delta$-gap profile, we require  agent $j_1$ is the only neighbor of seller $s$. Thus, for a normal agent, no matter how she reports she can not become the leading  agent. If a normal agent $i$
     becomes the winner after increasing her bid, then she will pay $H(a'_{-[i]})$. Note that agent $i$ is not the winner initially and thus $H(a'_{-[i]})\geq v_i$. Therefore, she will not get more utility by misreporting. For any case, no agents can get more utility by misreporting.

     Next, we consider a profile $a'$ that is not $\delta$-gap. For the case that the profile is still not $\delta$-gap after misreporting, it is the same for the IVCG and we have proved it. Next we assume that the profile will become $\delta$-gap after misreporting.
    For  agent $j_1$ , if she is still the leading agent after lowing down her bid to make the profile $\delta$-gap, then her utility will
     change from  $v_{j_1}$  to $H(a'_{-[j_2]})-\delta$. Note that according to the definition of $\delta$-gap profile,
     we have that $v_{j_1} \geq H(a'_{-[j_2]})-\delta$. Agent $j_1$ can not get more utility by misreporting.
     The other possible case is that a normal agent $i$ increases her bid to make the profile $\delta$-gap. If she is the winner after misreporting, she will pay $H(a'_{-[i]})$. Note that agent $i$ is not the winner initially and so    $H(a'_{-[i]})\geq v_i$. If she is not the winner, then she can only be a normal agent. Thus, she will not get more utility by misreporting.
     For any case, no agents can get more utility by misreporting.

     The mechanism always satisfies the IC property, regardless of the profile type.
  \end{proof}
  \noindent

On a profile with a leading agent $j_1$, the revenue in the IVCG is always $H(a'_{-[j_1]})$.
On profiles not $\delta$-gap, the revenue in the $\delta$-IVCG is the same as that in the IVCG;
on $\delta$-gap profiles, the revenue in $\delta$-IVCG is $\delta$ higher than that in the IVCG, where $\delta$ can be set as any positive value.
Compared to IVCG, the $\delta$-IVCG explores some descendants of the leading agent to get more revenue in some instances.
Due to the arbitrariness of `$\delta$', we can see that it is even harder to find a mechanism with the Pareto optimality of the revenue among all IR, WBB, and IC mechanisms.
The idea of $\delta$-IVCG may also be able to use in other mechanisms.

\section{False-Name Attacks}

  In our model,    the reported neighbor set $r'_i$   of each agent $i$ is a subset of the neighbour set $r_i$. We think that if $i$   reports  $i'$ then $i'\in r_i$ and agent $i$ knows $i'$. However, agent $i$ can  falsify her report by creating replicas of herself.
  In this way, the agent may   potentially gain more utility.  
To enhance the robustness of the mechanism, we examine false-name attacks, also known as ``Sybil attacks", which involve the creation of multiple replicas by a single agent \cite{shen2019multi}.

  \begin{defn}\textbf{(False-Name Attack)}
 The following action of an agent $i$ is called false-name attack. Agent $i$ creates  a set $RL_i$ of replicas,
 where agent in $RL_i\cup \{i\}$ can report any bid and any subset of $r_i\cup RL_i \cup \{i\}$.
 \end{defn}


  \begin{defn}
      An auction mechanism is \emph{false-name-proof (FNP)} if any agent can not get more utility by false-name attack.
  \end{defn}

 An agent $r$ may replace herself with an arbitrary subgraph in the spreading graph, which makes mechanisms hard to design  against false-name attacks.
 Thus, we also consider two simple   false-name attacks studied in the literature, e.g.,  see \cite{shen2019multi}.  In Type 1, all replicas are  parallel without any edges between them in the spreading graph,  and each of them can spread the sale information to a subset of the neighbors of $r$ (See Figure \ref{G2}(a)); In Type 2, all replicas are in series (Figure \ref{G2}(b)). By iteratively applying these two types and adding edges between replicas we may create complex false-name attacks. We can easily show that most mechanisms can not even handle even these two simple attacks.

  \begin{figure}[h]
      \centering
      \includegraphics[scale=0.2]{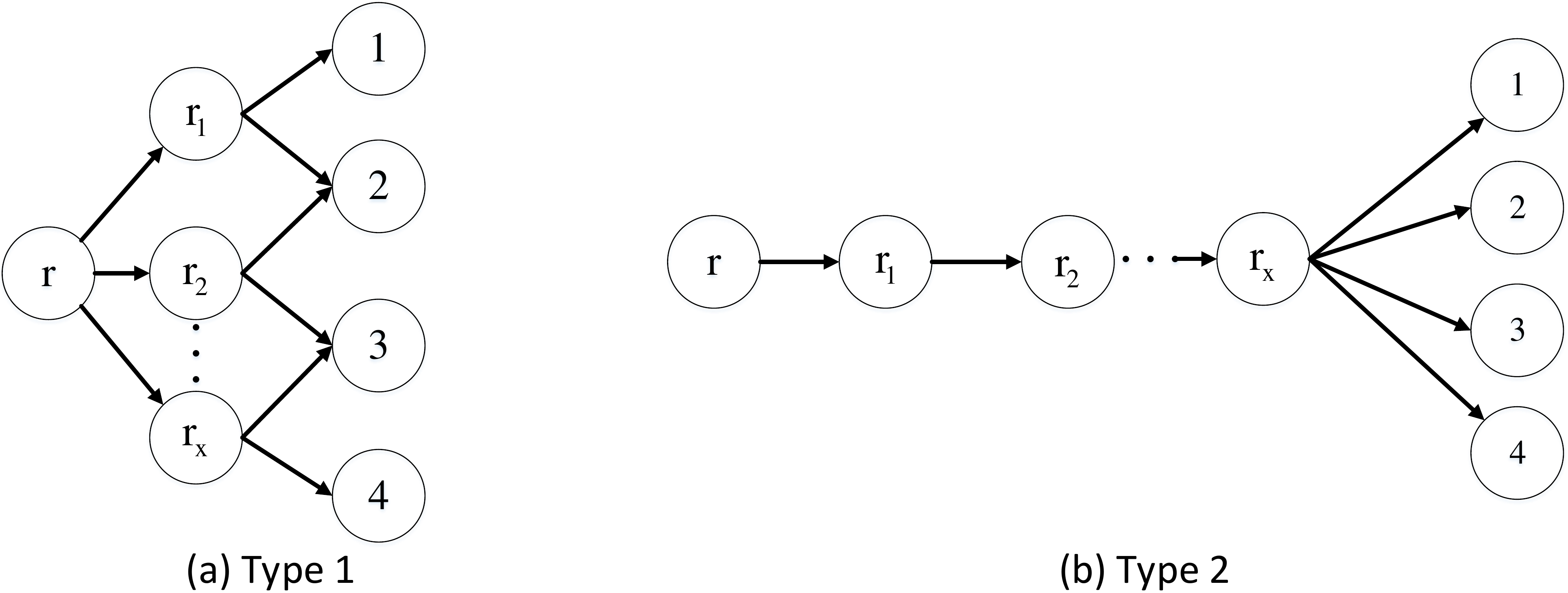}
      \caption{Two types of false-name attacks with replicas
       \label{G2}}
  \end{figure}

  \begin{figure}[h]
      \centering
      \includegraphics[width=0.38\textwidth]{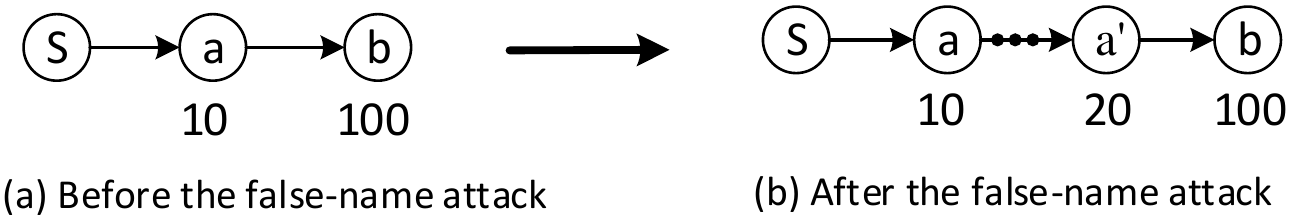}
      \caption{An example of Type 2 false-name attack}\label{fn1}
  \end{figure}

  \begin{figure}[h!]
      \centering
      \includegraphics[width=0.38\textwidth]{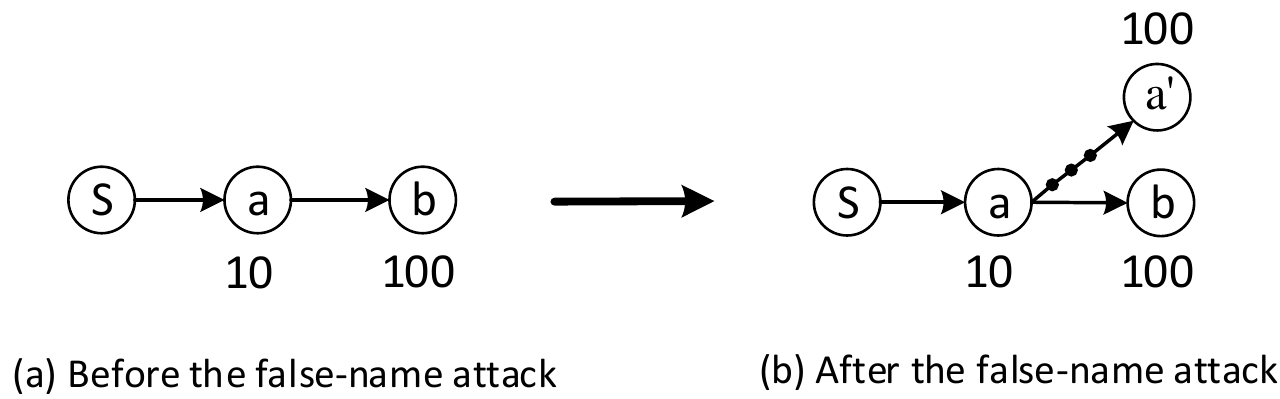}
      \caption{An example of Type 1 false-name attack}\label{fn2}
   \end{figure}

   \begin{lemma} \label{thm:VCG2}
  The VCG, PVCG, and $\delta$-IVCG mechanisms are not false-name-proof against Type 2 false-name attacks.
  \end{lemma}

  \begin{proof}
 Consider two agents, $a$ and $b$, with bids of 10 and 100, respectively. In the spreading graph shown in Figure \ref{fn1}(a), there exists an edge connecting seller
$s$ to agent $a$ and another edge connecting agents
$a$ and $b$.  By the payment policy of the VCG, agent $a$ gets $100$ in this example.   By the payment policy of the PVCG, agent $a$ gets $10$ in this example.  If agent $a$ creates a Type 2 replica $a'$ with the bid $20$ as shown in Figure \ref{fn1}(b), then agents $a$ and $a'$ get 100 and 90 receptively under the VCG, and agents $a$ and $a'$ get 10 and 10 receptively under the PVCG. The utility of agent $a$ will increase.   Hence, the VCG and PVCG are not false-name-proof.

  Next, we consider the $\delta$-IVCG. Let $a'$ be a $\delta$-gap profile, $j_1$ be the leading agent in it, and $j_2$ be the winner (a critical descendent of $j_1$ and also an interruption agent).
  The utility of $j_1$ in $\delta$-IVCG on $a'$ is $H(a'_{-[j_2]})-\delta$.
  The leading agent $j_1$ can create a Type 2 replica $j_2'$ with a bid $H(a'_{-[j_2]})< v'_{j'_2} < v'_{j_2}$ and then her utility will become $v'_{j'_2}-\delta$, which is larger than the previous utility.
  Hence, the $\delta$-IVCG is not false-name-proof.
  \end{proof}
  \noindent

 \begin{lemma}\label{thm:IDM2}
 The IDM is not false-name-proof against Type~1 false-name attacks.
 \end{lemma}
 \begin{proof}
     We show that the IDM is not false-name-proof against Type 1 attack. Before the false-name attack, the spreading graph is shown in Figure \ref{fn2}(a).  In the false-name attack, agent $a$ creates a Type 1 replica $a'$ and $a$ reports one more neighbor $a'$ as shown in Figure \ref{fn2}(b), where the bidding of $a'$ is 100. In the IDM, the utility of $a$ is $10$ if there is no false-name attack. In the case of the false name attack as described above, the utility of $a$ and $a'$ will be $100+0=100$.  Hence,  agent $a$ gets more utility by using this false-name attack.

    \end{proof}

  \begin{theo} \label{thm:FNP}
  The IVCG possesses the false-name-proof property.
  \end{theo}

  \begin{proof}
  For the case that there is a leading agent, if a non-leading agent creates a leading agent by the false-name-proof, she will pay money more than her truthful bid and get a negative utility.
  For the leading agent $i$, false-name attacks can not decrease $H(a'_{-[i]})$ and then she can not get more utility. The case that there is no leading agent is also easy. Thus, no agent can get more utility by false-name attacks under the IVCG mechanism.
  \end{proof}

  \section{Discussions}

We deeply analyzed how an agent influences the allocation efficiency through  different kinds of misreports and defined new  concepts such as  rewards, participation  rewards, and interruption agents.
With the help of these concepts, we can characterize mechanisms in network auctions, proved general theorems about the interplay among WBB, IR, IC, WIC, and efficiency,
design and simplify previous mechanisms.
We believe that these concepts are interesting and worthy of further study to better understand auctions in social networks.

\section*{Acknowledgments}
The work is supported by the National Natural Science Foundation of China, under grants 61972070.

\bibliography{ecai}
\end{document}